%% file: main.tex
\newif\ifels
\elsfalse
\ifels
	\documentclass[1p, sort&compress]{elsarticle}
	\newproof{proof}{Proof}
	\newtheorem{theorem}{Theorem}
	\journal{Applied and Computational Harmonic Analysis}
	\usepackage{lineno}
	\usepackage{color}
\else
	\documentclass[10pt, a4paper, onecolumn, twoside]{amsart}
	\input{misc/format}
	\usepackage{cite} 
	\newtheorem{theorem}{Theorem}[section]
\fi
\usepackage{amsmath} \usepackage{algorithm, algpseudocode}
\usepackage{amssymb, bm, graphicx, hyperref, enumerate, rotating}
\usepackage{verbatim, dsfont} 
\usepackage{subfigure}
\usepackage{epstopdf}
\input{misc/macro+variables}
\newcommand{\longtitle}{{Optimized Algorithms to Sample Determinantal Point Processes: \\ A Technical Report}}
\newcommand{\NTlong}{{Nicolas Tremblay}}
\newcommand{\SBlong}{{Simon Barthelm\'e}}
\newcommand{\POAlong}{{Pierre-Olivier Amblard}}
\newcommand{\NTshort}{{N.~Tremblay}}
\newcommand{\SBshort}{{S.~Barthelm\'e}}
\newcommand{\POAshort}{{P-O.~Amblard}}
\newcommand{\CNRS}{{CNRS, Univ Grenoble-Alpes, Gipsa-lab, France}}
\newcommand{\funds}{{This work was partly funded by LabEx PERSYVAL-Lab (ANR-11-LABX-0025-01), 
ANR GenGP (ANR-16-CE23-0008), Grenoble Data Institute and LIA CNRS/Melbourne Univ Geodesic.}}
\graphicspath{{figs/}}
\ifels
	\title{\longtitle\tnoteref{t1}}
	\tnotetext[t1]{\funds.}
	\author[cnrs]{\NTlong}
	\author[cnrs]{\SBlong}
	\author[cnrs]{\POAlong}
	\address[cnrs]{\CNRS}
\else
	\title[\longtitle]{\longtitle}
	\author[\NTshort]{\NTlong}
	\author[\SBshort]{\SBlong}
	\author[\POAshort]{\POAlong}
	\thanks{All three authors are with \CNRS. \funds}
\fi

\begin{document}

\ifels
	\input{manuscript/abstract}
	\maketitle
	\input{manuscript/core}

	\section*{Acknowledgements}
	\funds. 
\else
	\maketitle

\input{manuscript/abstract}\input{manuscript/core}
\fi

\ifels
	\bibliographystyle{elsarticle-num}
	\section*{References}
\else
	\bibliographystyle{IEEEtran}
\fi
\bibliography{biblio}

\end{document}

%% file: misc/format.tex
\usepackage{titlesec, color, amsthm}
\definecolor{blue1}{rgb}{0,0,0}
\renewcommand\thesection{\arabic{section}}
\renewcommand{\theequation}{\arabic{equation}}

\titleformat{\section}[hang]{\color{blue1}\large\bfseries\sffamily}{\thesection}{0mm}{. }[]
\titleformat{\subsection}[hang] {\color{blue1}\bfseries\sffamily}{\thesubsection}{0em}{. }[]
\titleformat{\subsubsection}[hang] {\color{blue1}\sffamily}{\thesubsubsection}{0em}{. }[]
\titlespacing*{\section}{1em}{3.5ex plus .2ex minus .2ex}{1ex plus .2ex}
\titlespacing*{\subsection}{0em}{3ex plus .2ex minus .2ex}{1ex plus .2ex}
\titlespacing*{\subsubsection}{0em}{3ex plus .2ex minus .2ex}{1ex plus .2ex}

\renewenvironment{abstract}{{\color{blue1}\small\bfseries Abstract.}\footnotesize}{\par \vskip .1in}
\makeatletter
\def\@setauthors{
\begingroup 
\def \thanks{\protect\thanks@warning}
\trivlist \centering\footnotesize \@topsep30\p@\relax \advance\@topsep by -\baselineskip
\item\relax \author@andify \authors \def\\{\protect\linebreak} {\color{blue1}\large\authors} \endtrivlist \endgroup}
\def\@settitle{\centering{\color{blue1} \Large \bfseries \bfseries \@title \par}}
\makeatother
\usepackage[normal, small, up, textfont=it]{caption}

%% file: misc/macro+variables.tex
\newcommand{\fou}{\ensuremath{\vec{u}}}

\renewcommand{\leq}{\ensuremath{\leqslant}}
\renewcommand{\geq}{\ensuremath{\geqslant}}

\newcommand{\adjoint}{\ensuremath{{\intercal}}}

\newcommand{\norm}[1]{\ensuremath{\left\| #1\right\|}}

\newcommand{\ma}[1]{\ensuremath{\mathsf{#1}}}
\renewcommand{\vec}[1]{\ensuremath{\bm{#1}}}
\newcommand{\diag}{\ensuremath{{\rm diag}}}

\newcommand{\ie}{\textit{i.e.}}

\newtheorem{definition}[theorem]{Definition}
\newtheorem{lemma}[theorem]{Lemma}

\algnewcommand\algorithmicinput{\textbf{Input:}}
\algnewcommand\Input{\item[\algorithmicinput]}

\algnewcommand\algorithmicoutput{\textbf{Output:}}
\algnewcommand\Output{\item[\algorithmicoutput]}

%% file: manuscript/abstract.tex
\begin{abstract}
  In this technical report, we discuss several sampling algorithms for Determinantal Point  Processes (DPP). DPPs have recently gained a broad interest in the machine learning and statistics literature as random point processes with negative correlation, \ie, ones that can generate a "diverse" sample from a set of items. They are parametrized by a matrix $\ma{L}$, called $L$-ensemble, that encodes the correlations between items. The standard sampling algorithm is separated in three phases: 1/~eigendecomposition of $\ma{L}$, 2/~an eigenvector sampling phase where $\ma{L}$'s eigenvectors are sampled independently via a Bernoulli variable parametrized by their associated eigenvalue, 3/~a Gram-Schmidt-type orthogonalisation procedure of the sampled eigenvectors.
  In a naive implementation, the computational cost of the third step is on average $\mathcal{O}(N\mu^3)$ where $\mu$ is the average number of samples of the DPP. We give an algorithm which runs in $\mathcal{O}(N\mu^2)$ and is extremely simple to implement. If memory is a constraint, we also describe a dual variant with reduced memory costs. In addition, we discuss implementation details often missing in the literature. 
\end{abstract}

%% file: manuscript/core.tex
\section{Introduction}
\label{sec:intro}

Determinantal Point Processes enable a form of subsampling that generalises sampling without replacement: from a set of items $\mathcal{X}$, we draw a random subset $\mathcal{S} \subseteq \mathcal{X}$ such that $\mathcal{S}$ preserves some of the diversity in $\mathcal{X}$.
Here we focus solely on discrete DPPs, where $\mathcal{X}$ has $N$ elements. A generic algorithm for exact sampling from discrete DPPs was given in \cite{hough_determinantal_2006}, and popularised by \cite{kulesza_determinantal_2012}. Implemented naively, this algorithm has cost $\mathcal{O}(N\mu^{3})$. The point of this note is to derive and describe a simpler algorithm with cost $\mathcal{O}(N\mu^{2})$. We make no great claim to novelty, as other algorithms with the same asymptotic cost exist, but the one we give is very easily stated and trivial to implement.

\subsection{Notations}
 Sets are in upper-case calligraphic letters: $\mathcal{S}, \mathcal{K}, \ldots$. Vectors are in lower-case bold letters: $\vec{x}, \vec{f}, \ldots$.  The $i$-th entry of vector $\vec{x}$ is written either $x_i$ or $x(i)$. For instance, $\vec{\delta}_i$ is the vector such that: $\forall j\neq i, \delta_i(j)=0$ and $\delta_i(i)=1$.   Matrices are in upper case letters: $\ma{M}, \ma{U}, \ldots$. For instance $\ma{I}_d\in\mathbb{R}^{d\times d}$ is the identity matrix in dimension $d$. If the dimension is not specified, it can be guessed via the context. The $(i,j)$-th element of matrix $\ma{M}$ is written $\ma{M}_{i,j}$ and its $j$-th column is written $\vec{m}_j$. $\ma{M}_{\mathcal{A},\mathcal{B}}$ is the restriction of matrix $\ma{M}$ to the rows (resp. columns) indexed by the elements of $\mathcal{A}$ (resp. $\mathcal{B}$). We write $\ma{M}_{\mathcal{A}}$ as a shorthand for $\ma{M}_{\mathcal{A},\mathcal{A}}$. Moreover, $[N]$ stands for the set of $N$ first integers $\{1,2,\ldots, N\}$. Finally, the notation $\det(\ma{M})$ stands for the determinant of matrix $\ma{M}$. 

\subsection{Definitions}
We give here a definition of DPPs via $L$-ensembles. Equivalent formulations are based on the marginal kernel and specify marginal probabilities (see \emph{e.g.}~\cite{kulesza_determinantal_2012} for details).
\begin{definition}[Determinantal Point Process] 
	\label{def:DPP} Consider a point process, \ie, a process that randomly draws an element $\mathcal{S}\in[N]$. It is determinantal if there exists a semi-definite positive (SDP) matrix $\ma{L}\in\mathbb{R}^{N\times N}$ such that the probability of sampling $\mathcal{S}$ is:
	$$\mathbb{P}(\mathcal{S}) = \frac{\det(\ma{L}_{\mathcal{S}})}{\det(\ma{I}+\ma{L})}.$$
	$\ma{L}$ is called the $L$-ensemble of the DPP.
\end{definition}
$\ma{L}$ being SDP it is diagonalisable in:
\begin{align}
\ma{L} = \ma{U\Lambda U}^\top,
\end{align}
with $\ma{U}=(\vec{u}_1|\vec{u}_2|\ldots|\vec{u}_N)\in\mathbb{R}^{N\times N}$ its set of orthonormal eigenvectors and $\ma{\Lambda}=\diag(\lambda_1|\ldots|\lambda_N)\in\mathbb{R}^{N\times N}$ the diagonal matrix of sorted eigenvalues: $0\leq\lambda_1\leq\ldots\leq\lambda_N$. 

The number of samples $|\mathcal{S}|$ of a DPP is random and is distributed according to a sum of $N$ Bernoulli variables of parameters $\lambda_n/(1+\lambda_n)$~\cite{kulesza_determinantal_2012}. The expected number of samples is thus $$\mu=\mathbb{E}(|\mathcal{S}|)=\sum_n \frac{\lambda_n}{1+\lambda_n}$$ and the variance of the number of samples is 
$$\text{Var}(|\mathcal{S}|) = \sum_n \frac{\lambda_n}{(1+\lambda_n)^2}.$$

In many cases, it is preferable to constrain the DPP to output a fixed number of samples $k$. This leads to $k$-DPPs:

\begin{definition}[$k$-DPP~\cite{kulesza_determinantal_2012}] 
	\label{def:kDPP} Consider a point process that randomly draws an element $\mathcal{S}\in[N]$. This process is a $k$-DPP with $L$-ensemble $\ma{L}$ if:
	\begin{enumerate}
		\item  $\forall\mathcal{S}\text{ s.t. } |\mathcal{S}|\neq k, ~ \mathbb{P}(\mathcal{S}) = 0$
		\item $\forall\mathcal{S}\text{ s.t. } |\mathcal{S}|= k, ~ \mathbb{P}(\mathcal{S}) = \frac{1}{Z} \det(\ma{L}_{\mathcal{S}})$, with $Z$ the constant s.t.  $\displaystyle\frac{1}{Z}\sum_{\mathcal{S}\text{ s.t. } |\mathcal{S}|=k}\det(\ma{L}_{\mathcal{S}})=1$.
	\end{enumerate}
\end{definition}

\subsection{The standard DPP sampling algorithm}
The standard algorithm to sample a DPP from a $L$-ensemble given its eigendecomposition may be decomposed in two steps~\cite{hough_determinantal_2006}:
\begin{enumerate}
	\item[i/] Sample eigenvectors. Draw $N$ Bernoulli variables with parameters $\lambda_n/(1+\lambda_n)$: for $n=1,\ldots,N$, add $n$ to the set of sampled indices $\mathcal{K}$ with probability $\lambda_n/(1+\lambda_n)$. We generically denote by $k$ the number of elements in $\mathcal{K}$. Note that the expected value of $k$ is $\mu$.
	\item[ii/] Run alg.~\ref{alg:sampling_DPP_standard} to sample a $k$-DPP with projective $L$-ensemble $\ma{P}=\ma{V} \ma{V}^\adjoint$ where  $\ma{V}\in\mathbb{R}^{N\times k}$ concatenates all eigenvectors $\fou_n$ such that $n\in\mathcal{K}$. Note that $\ma{V}^\adjoint \ma{V}=\ma{I}_k$. 
\end{enumerate}

For the proof that the combination of these two steps samples a DPP with $L$-ensemble $\ma{L}$, we refer the reader to the papers~\cite{hough_determinantal_2006, kulesza_determinantal_2012}. 

\begin{algorithm}
	\caption{Standard $k$-DPP sampling algorithm with projective $L$-ensemble $\ma{P}=\ma{V} \ma{V}^\adjoint$}
	\label{alg:sampling_DPP_standard}
	\begin{algorithmic}
		\Input $\ma{V}\in\mathbb{R}^{N\times k}$ such that $\ma{V}^\adjoint \ma{V}=\ma{I}_k$\\
		$\mathcal{S} \leftarrow \emptyset$\\
		\textbf{for} $n=1,\ldots,k$ \textbf{do}:\\
		\hspace{0.5cm}$\bm{\cdot}$ Define $\vec{p}\in\mathbb{R}^N$ : $\forall i, \quad p(i) = \norm{\ma{V}^\adjoint\vec{\delta}_i}^2$.\\
		\hspace{0.5cm}$\bm{\cdot}$ Draw $s_n$ with probability $\mathbb{P}(s)=p(s)/\sum_{i}p(i)$\\
		\hspace{0.5cm}$\bm{\cdot}$ $\mathcal{S} \leftarrow \mathcal{S}\cup\{s_n\}$\\
		\hspace{0.5cm}$\bm{\cdot}$ Compute $\ma{V}_{\perp}\in\mathbb{R}^{N\times (k-n)}$, the orthonormal basis for the subspace spanned by the columns of\\
		\hspace{1cm} $\ma{V}$ and orthogonal to $\vec{\delta}_{s_n}$. \\
		\hspace{0.5cm}$\bm{\cdot}$ Update $\ma{V}\leftarrow\ma{V}_{\perp}$.\\
		\textbf{end for}
		\Output $\mathcal{S}$ of size $k$.
	\end{algorithmic}
\end{algorithm}

The expensive step in alg.~\ref{alg:sampling_DPP_standard} is the orthogonalisation, which could be sped up via low-rank updates to a QR or SVD decomposition. We suggest a more direct route, leading to algorithm  \ref{alg:sampling_DPP_efficient}. To derive this algorithm, we first introduce alg.~\ref{alg:sampling_DPP_alt} as a stepping stone to simplify the proofs, but readers only interested in the implementation can skip ahead to alg.~\ref{alg:sampling_DPP_efficient}.

\section{Improved sampling algorithm}
\subsection{Primal representation}

An alternative to alg.~\ref{alg:sampling_DPP_standard} to compute step ii/ above is given in alg.~\ref{alg:sampling_DPP_alt}. In fact:

\begin{algorithm}
	\caption{Alternative $k$-DPP sampling algorithm with projective $L$-ensemble $\ma{P}=\ma{V} \ma{V}^\adjoint$}
	\label{alg:sampling_DPP_alt}
	\begin{algorithmic}
		\Input $\ma{V}\in\mathbb{R}^{N\times k}$ such that $\ma{V}^\adjoint \ma{V}=\ma{I}_k$\\
		Write $\forall i, ~~\vec{y}_i=\ma{V}^\adjoint\vec{\delta_i}\in\mathbb{R}^{k}$.\\
		$\mathcal{S} \leftarrow \emptyset$\\
		Define $\vec{p}_0\in\mathbb{R}^N$ : $\forall i, \quad p_0(i) = \norm{\vec{y}_i}^2$\\
		$\vec{p} \leftarrow\vec{p}_0$\\
		\textbf{for} $n=1,\ldots,k$ \textbf{do}:\\
		\hspace{0.5cm}$\bm{\cdot}$ Draw $s_n$ with probability $\mathbb{P}(s)=p(s)/\sum_{i}p(i)$\\
		\hspace{0.5cm}$\bm{\cdot}$ $\mathcal{S} \leftarrow \mathcal{S}\cup\{s_n\}$\\
		\hspace{0.5cm}$\bm{\cdot}$ Compute $\ma{P}_\mathcal{S}=\ma{V}_{\mathcal{S},[k]} \ma{V}_{\mathcal{S},[k]}^\adjoint\in\mathbb{R}^{n\times n}$, its inverse $\ma{P}_\mathcal{S}^{-1}$ and $\ma{P}_{\mathcal{S},i} = \ma{V}_{\mathcal{S},[k]}\vec{y}_i\in\mathbb{R}^{n}$.\\
		\hspace{0.5cm}$\bm{\cdot}$ Update $\vec{p}$ :
		$\forall i\quad p(i) = p_0(i) - \ma{P}_{\mathcal{S},i}^\adjoint \ma{P}_\mathcal{S}^{-1} \ma{P}_{\mathcal{S},i}$\\
		\textbf{end for}
		\Output $\mathcal{S}$ of size $k$.
	\end{algorithmic}
\end{algorithm}

\begin{lemma}
	Alike alg.~\ref{alg:sampling_DPP_standard}, alg.~\ref{alg:sampling_DPP_alt} samples a $k$-DPP with projective $L$-ensemble $\ma{P}=\ma{V} \ma{V}^\adjoint$.
\end{lemma}

\begin{proof}
	Let us denote by $\mathcal{S}$ the output of alg.~\ref{alg:sampling_DPP_alt}. 
	Let us also denote by $\mathcal{S}_n$ (resp. $p_n(i)$) the sample set (resp. the value of $p(i)$) at the end of the $n$-th iteration of the loop of alg.~\ref{alg:sampling_DPP_alt}. We have : $\mathcal{S}_n = \mathcal{S}_{n-1}\cup\{s_n\}$. 
	Using the Schur complement, we have : 
	\begin{align}
	\label{eq:pn}
	\forall n\in[1,k]\;,\forall i\qquad\text{det}\left(\ma{P}_{\mathcal{S}_{n-1}\cup\{i\}}\right) &= \left(\ma{P}_{ii} - \ma{P}_{\mathcal{S}_{n-1},i}^\adjoint\ma{P}_{\mathcal{S}_{n-1}}^{-1}\ma{P}_{\mathcal{S}_{n-1},i}\right) \;\text{det}\left(\ma{P}_{\mathcal{S}_{n-1}}\right) \nonumber\\
	& = p_{n-1}(i)\;\text{det}\left(\ma{P}_{\mathcal{S}_{n-1}}\right).
	\end{align}
Given Eq.~\eqref{eq:pn}, and as $\forall\mathcal{S}, ~~\ma{P}_\mathcal{S}$ is SDP by construction (such that $\text{det}(\ma{P}_\mathcal{S})\geq0$), one can show that $p_n(i)\geq 0$ and $\sum_i{p_n}(i)\neq 0$: at each iteration $n$, the probability  $\mathbb{P}(s)=\frac{p_n(s)}{\sum_i p_n(i)}$ is thus well defined. 
	%
	The loop being repeated $k$ times, the number of samples of the output is thus necessarily equal to $k$. Thus, $\mathbb{P}(\mathcal{S})=0$ for all $\mathcal{S}$ of size different than $k$. 
	
	Let us now show that $\mathbb{P}(\mathcal{S})$ is indeed of determinantal form when $|\mathcal{S}|=k$. By construction of $\mathcal{S}$ :
	\begin{align}
	\mathbb{P}(\mathcal{S}) &= \prod_{n=1}^k \mathbb{P}(s_n|s_1, s_2,\ldots, s_{n-1})=\prod_{n=1}^k \frac{p_{n-1}(s_n)}{\sum_{i=1}^N p_{n-1}(i)}.\label{eq:prod_cond}
	\end{align}
	Writing Eq.~\eqref{eq:pn} for $i=s_n$, and iterating, one obtains  :
	$\prod_{n=1}^k p_{n-1}(s_n)=\text{det}(\ma{P}_{\mathcal{S}})$. 
	Let us finish by showing that the denominator of Eq~\eqref{eq:prod_cond} does not depend on the chosen samples. This is where the projective assumption of $\ma{P}$ is essential. One has :
	\begin{align}
	\forall n\in[1,k], \qquad \sum_{i=1}^N p_{n-1}(i) = \sum_{i=1}^N p_{0}(i) - \sum_{i=1}^N \ma{P}_{\mathcal{S}_{n-1},i}^\adjoint\ma{P}_{\mathcal{S}_{n-1}}^{-1}\ma{P}_{\mathcal{S}_{n-1},i}\nonumber
	\end{align}
	We have $\sum_{i=1}^N p_{0}(i)=\sum_{i=1}^N\norm{\vec{y}_i}^2=\text{Tr}(\ma{VV}^\adjoint)=\text{Tr}(\ma{V}^\adjoint\ma{V})=k$. 
	Moreover, by invariance of the trace to circular permutations, we have:
	\begin{align*}
	\sum_{i=1}^N \ma{P}_{\mathcal{S}_{n-1},i}^\adjoint\ma{P}_{\mathcal{S}_{n-1}}^{-1}\ma{P}_{\mathcal{S}_{n-1},i} &= \text{Tr}\left(\ma{V}\ma{V}_{\mathcal{S}_{n-1},[k]}^\adjoint\ma{P}_{\mathcal{S}_{n-1}}^{-1}\ma{V}_{\mathcal{S}_{n-1},[k]}\ma{V}^\adjoint\right)\\
	&= \text{Tr}\left(\ma{P}_{\mathcal{S}_{n-1}}^{-1}\ma{V}_{\mathcal{S}_{n-1},[k]}\ma{V}^\adjoint\ma{V}\ma{V}_{\mathcal{S}_{n-1},[k]}^\adjoint\right)\\
	&=
	\text{Tr}\left(\ma{P}_{\mathcal{S}_{n-1}}^{-1}\ma{V}_{\mathcal{S}_{n-1},[k]}\ma{V}_{\mathcal{S}_{n-1},[k]}^\adjoint\right)\\
	&=
	\text{Tr}\left(\ma{P}_{\mathcal{S}_{n-1}}^{-1}\ma{P}_{\mathcal{S}_{n-1}}\right)\\
	&= \text{Tr}(\ma{I}_{n-1}) = n-1,
	\end{align*}
	Thus 
	\begin{align}
	&\forall\mathcal{S}\text{ s.t. } |\mathcal{S}|\neq k, ~ \mathbb{P}(\mathcal{S}) = 0\\
	\text{and }~~\qquad&\forall\mathcal{S}\text{ s.t. } |\mathcal{S}|= k, ~ \mathbb{P}(\mathcal{S}) = \frac{1}{Z} \det(\ma{P}_{\mathcal{S}})\text{~with~} Z = \prod_{n=1}^k k-n+1=k\,!
	\end{align}
	which ends the proof.
\end{proof}

Note that both Algorithms~\ref{alg:sampling_DPP_standard} and~\ref{alg:sampling_DPP_alt} have a computation cost of order $\mathcal{O}(Nk^3)$ such that the overall sampling algorithm given the eigendecomposition of $\ma{L}$ is in average of the order $\mathcal{O}(N\mu^3)$. 
This is suboptimal. In fact, the scalar $\ma{P}_{\mathcal{S},i}^\adjoint \ma{P}_\mathcal{S}^{-1} \ma{P}_{\mathcal{S},i}$ is computed from scratch at each iteration of the loop even though one could use Woodbury's identity to take advantage of computations done at past iterations. This observation leads to alg.~\ref{alg:sampling_DPP_efficient}. We have the following equivalence:

\begin{lemma}
	Alg.~\ref{alg:sampling_DPP_efficient} is equivalent to  alg.~\ref{alg:sampling_DPP_alt}: it also samples a $k$-DPP with projective $L$-ensemble $\ma{P}=\ma{V} \ma{V}^\adjoint$.
\end{lemma}

\begin{algorithm}
	\caption{Efficient $k$-DPP sampling algorithm with projective $L$-ensemble $\ma{P}=\ma{V} \ma{V}^\adjoint$}
	\label{alg:sampling_DPP_efficient}
	\begin{algorithmic}
		\Input $\ma{V}\in\mathbb{R}^{N\times k}$ such that $\ma{V}^\adjoint \ma{V}=\ma{I}_k$\\
		Write $\forall i, ~~\vec{y}_i=\ma{V}^\adjoint\vec{\delta_i}\in\mathbb{R}^{k}$.\\
		$\mathcal{S} \leftarrow \emptyset$\\
		Define $\vec{p}\in\mathbb{R}^N$ : $\forall i, \quad p(i) = \norm{\vec{y}_i}^2$\\
		\textbf{for} $n=1,\ldots,k$ \textbf{do}:\\
		\hspace{0.5cm}$\bm{\cdot}$ Draw $s_n$ with prob. $\mathbb{P}(s)=p(s)/\sum_{i}p(i)$\\
		\hspace{0.5cm}$\bm{\cdot}$ $\mathcal{S} \leftarrow \mathcal{S}\cup\{s_n\}$\\
		\hspace{0.5cm}$\bm{\cdot}$ Compute 
		$\vec{f}_n = \vec{y}_{s_n} - \sum_{l=1}^{n-1} \vec{f}_l(\vec{f}_l^\adjoint\vec{y}_{s_n})\in\mathbb{R}^{k}$\\
		\hspace{0.5cm}$\bm{\cdot}$ Normalize $\vec{f}_n \leftarrow \vec{f}_n / \sqrt{\vec{f}_n^\adjoint\vec{y}_{s_n}}$\\
		\hspace{0.5cm}$\bm{\cdot}$ Update $\vec{p}$ : $\forall i\quad p(i) \leftarrow p(i) - (\vec{f}_n^\adjoint\vec{y}_i)^2$\\
		\textbf{end for}
		\Output $\mathcal{S}$ of size $k$.
	\end{algorithmic}
\end{algorithm}

\begin{proof}
	Let us denote by $\mathcal{S}_n$ (resp. $p_n(i)$) the sample set (resp. the value of $p(i)$) at the end of the $n$-th iteration of the loop. Let us also denote by $p_0(i)$ the initial value of $p(i)$. All we need to show is that the $p_n(i)$ are equal in both algorithms. In alg.~\ref{alg:sampling_DPP_efficient}: $p_n(i) = p_{n-1}(i) - (\vec{f}_n^\adjoint\vec{y}_i)^2 = p_0(i) - \sum_{l=1}^{n} (\vec{f}_l^\adjoint\vec{y}_i)^2$ (where the vectors $\vec{f}_l$ and $\vec{y}_i$ are defined in the algorithm). Comparing with the  $p_n(i)$ of alg.~\ref{alg:sampling_DPP_alt}, all we need to show is:
	\begin{align}
	\forall n\forall i \quad \quad
	\sum_{l=1}^{n} (\vec{f}_l^\adjoint\vec{y}_i)^2=  \ma{P}_{\mathcal{S}_n,i}^\adjoint \ma{P}_{\mathcal{S}_n}^{-1} \ma{P}_{\mathcal{S}_n,i}.
	\end{align}
	We will show more generally that:
	\begin{align}
	\label{eq:to_proove_primal}
	\forall n, \forall(i,j)  \quad \quad
	\sum_{l=1}^{n} (\vec{f}_l^\adjoint \vec{y}_i) (\vec{f}_l^\adjoint \vec{y}_j) =  \ma{P}_{\mathcal{S}_n,i}^\adjoint \ma{P}_{\mathcal{S}_n}^{-1} \ma{P}_{\mathcal{S}_n,j}.
	\end{align}
	To do so, we propose a recurrence. Before we start, let us note that, by definition:
	\begin{align}
	\forall (i,j)\qquad \vec{y}_i^\adjoint\vec{y}_j=\vec{\delta}_i^\adjoint\ma{VV}^\adjoint\vec{\delta}_j=\vec{\delta}_i^\adjoint\ma{P}\vec{\delta}_j=\ma{P}_{ij}.
	\end{align}\\
	\textit{Initialisation}. It is true for $n=1$, where $\mathcal{S}_1$ is reduced to $\{s_1\}$  and:
	\begin{align}
	\ma{P}_{\mathcal{S}_1}^{-1} =  \frac{1}{\ma{P}_{s_1,s_1}}
	\end{align}
	is a scalar. Indeed, the following holds for all $(i,j)$:
	\begin{align}
	(\vec{f}_1^\adjoint\vec{y}_i) (\vec{f}_1^\adjoint \vec{y}_j) = \frac{(\vec{y}_{s_1}^\adjoint\vec{y}_i)(\vec{y}_{s_1}^\adjoint \vec{y}_j)}{\norm{\vec{y}_{s_1}}^2} = \frac{\ma{P}_{i,s_1}\ma{P}_{s_1,j}}{\ma{P}_{s_1,s_1}} = \ma{P}_{\mathcal{S}_1,i}^\adjoint \ma{P}_{\mathcal{S}_1}^{-1} \ma{P}_{\mathcal{S}_1,j}.
	\end{align}
	\textit{Hypothesis}. We assume that Eq.~\eqref{eq:to_proove_primal} is true at iteration $n-1$.\\
	\textit{Recurrence}. Let us show it is also true at iteration $n$. Using Woodbury's identity on $\ma{P}_{\mathcal{S}_n}^{-1}$, we show that:
	 \begin{align}
	 \label{eq:en_passant}
	\ma{P}_{\mathcal{S}_n,i}^\adjoint \ma{P}_{\mathcal{S}_n}^{-1} \ma{P}_{\mathcal{S}_n,j} &=
	\ma{P}_{\mathcal{S}_{n-1},i}^\adjoint \ma{P}_{\mathcal{S}_{n-1}}^{-1} \ma{P}_{\mathcal{S}_{n-1},j} + \frac{z_n(i)z_n(j)}{z_n(s_n)},
	\end{align}
	where $z_n(i)=\ma{P}_{s_{n},i}-\ma{P}_{\mathcal{S}_{n-1},s_n}^\adjoint \ma{P}_{\mathcal{S}_{n-1}}^{-1}\ma{P}_{\mathcal{S}_{n-1},i}$.
	Applying the hypothesis to $\ma{P}_{\mathcal{S}_{n-1},i}^\adjoint \ma{P}_{\mathcal{S}_{n-1}}^{-1} \ma{P}_{\mathcal{S}_{n-1},j}$ in Eq.~\eqref{eq:en_passant}, the proof boils down to showing that:
	\begin{align}
	\label{eq:reste_a_montrer_primal}
	\forall i,j \quad \quad (\vec{f}_n^\adjoint\vec{y}_i) (\vec{f}_n^\adjoint \vec{y}_j) = \frac{z_n(i)z_n(j)}{z_n(s_n)}.
	\end{align}
	By construction of Algorithm~\ref{alg:sampling_DPP_efficient}, $\vec{f}_n^\adjoint\vec{y}_i$ reads :
	 \begin{align}
	\label{eq:f_primal}
	\forall i\quad\quad \vec{f}_n^\adjoint\vec{y}_i = \frac{\ma{P}_{s_{n},i} - \sum_{l=1}^{n-1}(\vec{f}_l^\adjoint \vec{y}_i) (\vec{f}_l^\adjoint \vec{y}_{s_n})}{\sqrt{\ma{P}_{s_{n},s_n} - \sum_{l=1}^{n-1}(\vec{f}_l^\adjoint \vec{y}_{s_n})^2}}.
	\end{align}
	Applying once more the hypothesis on $\sum_{l=1}^{n-1}(\vec{f}_l^\adjoint \vec{y}_i) (\vec{f}_l^\adjoint \vec{y}_{s_n})$ and $\sum_{l=1}^{n-1}(\vec{f}_l^\adjoint \vec{y}_{s_n})^2$, one obtains:
	\begin{align}
	\forall i\quad\quad 
	\vec{f}_n^\adjoint\vec{y}_i = \frac{z_n(i)}{\sqrt{z_n(s_n)}},\vspace{-0.3cm}
	\end{align}
	which shows Eq.~\eqref{eq:reste_a_montrer_primal} and ends the proof.
\end{proof}

Alg.~\ref{alg:sampling_DPP_efficient} samples a $k$-DPP with a computation cost of order $\mathcal{O}(Nk^2)$ instead of $\mathcal{O}(Nk^3)$, thereby gaining an  order of magnitude. Moreover, it is straightforward to implement.

\subsection{Dual (low-rank) representation}

The gain obtained in the previous section is not significant in the general case, as the limiting step of the overall sampling algorithm is any case the diagonalisation of $\ma{L}$, that costs $\mathcal{O}(N^3)$. Thankfully, in many  applications, a dual representation of $\ma{L}$  exists, \ie, a representation of $\ma{L}$ in a low-rank form 
\begin{align}
\ma{L}=\ma{\Psi}^\adjoint\ma{\Psi},
\end{align}
where $\ma{\Psi}=(\vec{\psi}_1|\ldots|\vec{\psi_N})\in\mathbb{R}^{d\times N}$ with $d$ a dimension that can be significantly smaller than $N$. In this case, we will see that preferring alg.~\ref{alg:sampling_DPP_efficient} to alg.~\ref{alg:sampling_DPP_standard} does induce a significant gain in the overall  computation time. 
The dual representation enables us to circumvent the $\mathcal{O}(N^3)$ diagonalization cost of $\ma{L}$ and only diagonalize the dual form:
\begin{align}
\ma{C} = \ma{\Psi}\ma{\Psi}^\adjoint  \in\mathbb{R}^{d\times d},
\end{align}
costing only $\mathcal{O}(Nd^2)$ (time to compute $\ma{C}$ from $\ma{\Psi}$ and to compute the low-dimensional diagonalization). $\ma{C}$'s eigendecomposition yields:
\begin{align}
\ma{C} = \ma{R}\ma{E}\ma{R}^\adjoint,
\end{align}
with $\ma{R}=(\vec{r}_1|\ldots|\vec{r}_{d})\in\mathbb{R}^{d\times d}$ the basis of eigenvectors and $\ma{E}\in\mathbb{R}^{d\times d}$ the diagonal matrix of eigenvalues such that $0\leq e_1\leq\ldots\leq e_{d}$. 
%
One can show (e.g., see Proposition 3.1 in~\cite{kulesza_determinantal_2012}) that all eigenvectors associated to non-zero eigenvalues of $\ma{L}$ can be recovered from $\ma{C}$'s eigendecomposition. More precisely, if $\vec{r}_k$ is an eigenvector of $\ma{C}$ associated to eigenvalue $e_k$, then:
\begin{align}
\vec{u}_k = \frac{1}{\sqrt{e_k}} \ma{\Psi}^\adjoint \vec{r}_k
\end{align}
is a normalized eigenvector of $\ma{L}$ with same eigenvalue. Thus, given the eigendecomposition of the dual form, the standard algorithm boils down to:
\begin{enumerate}
	\item[i/] Sample eigenvectors. Draw $N$ Bernoulli variables with parameters $e_n/(1+e_n)$: for $n=1,\ldots,d$, add $n$ to the set of sampled indices $\mathcal{K}$ with probability $e_n/(1+e_n)$. Denote by $k$ the number of elements of $\mathcal{K}$. $k$ is necessarily smaller than $d$. 
	\item[ii/] Denote by $\ma{W}\in\mathbb{R}^{d\times k}$ the matrix concatenating all eigenvectors $\vec{r}_n$ such that $n\in\mathcal{K}$, and $\tilde{\ma{E}}=\diag(\{e_n\}_{n\in\mathcal{K}})\in\mathbb{R}^{k\times k}$ the diagonal matrix of associated eigenvalues. Run alg.~\ref{alg:sampling_DPP_standard},~\ref{alg:sampling_DPP_alt} or~\ref{alg:sampling_DPP_efficient}  to sample a $k$-DPP with projective $L$-ensemble $\ma{P}=\ma{V} \ma{V}^\adjoint$ where  $\ma{V}=\ma{\Psi}^\adjoint\ma{W}\tilde{\ma{E}}^{-1/2}\in\mathbb{R}^{N\times k}$ concatenates the reconstructed eigenvectors in dimension $N$. Note that $\ma{V}^\adjoint \ma{V}=\ma{I}_k$. 
\end{enumerate}
The reconstruction operation $\ma{V}=\ma{\Psi}^\adjoint\ma{W}\tilde{\ma{E}}^{-1/2}$ costs $\mathcal{O}(Ndk)$. Thus, preferring alg.~\ref{alg:sampling_DPP_efficient} to alg.~\ref{alg:sampling_DPP_standard} in step ii/ lowers the total computation time of sampling a DPP from $\mathcal{O}(N(\mu^3+d^2)$ in average to $\mathcal{O}(Nd^2)$ (as $d$ is necessarily larger than $\mu$). 

\textbf{A last algorithm in case of memory issues.} If, for memory optimization reasons, one does not desire to reconstruct the eigenvectors in dimension $N$, one may slightly alter alg.~\ref{alg:sampling_DPP_efficient} by noticing that 
$\vec{y}_i$, in the first line of the algorithm, reads:
\begin{align}
\label{eq:yi}
	\vec{y}_i=\ma{V}^\adjoint\vec{\delta}_i = \tilde{\ma{E}}^{-1/2}\ma{W}^\adjoint\ma{\Psi}\vec{\delta}_i = \tilde{\ma{E}}^{-1/2}\ma{W}^\adjoint\vec{\psi}_i.
\end{align}
One may thus first precompute $\tilde{\ma{C}}=\ma{W}\tilde{\ma{E}}^{-1}\ma{W}^\adjoint$, then directly work with $\vec{\psi}_i$ instead of $\vec{y}_i$ and replace 1)~$\norm{\vec{y}_i}^2$ by $\vec{\psi}_i^\adjoint \tilde{\ma{C}}\vec{\psi}_i$ and 2)~all scalar products of the type $\vec{f}_n^\adjoint\vec{y}_i$  by $=\vec{f}_n^\adjoint\tilde{\ma{C}}\vec{y}_i$. This leads to alg.~\ref{alg:sampling_DPP_efficient_dual}. This algorithm is nevertheless slightly heavier computationally than alg.~\ref{alg:sampling_DPP_efficient}: it indeed runs in $\mathcal{O}(Ndk)$ instead of $\mathcal{O}(Nk^2)$. 

\begin{algorithm}
	\caption{Efficient $k$-DPP sampling algorithm in case of a dual representation.}
	\label{alg:sampling_DPP_efficient_dual}
	\begin{algorithmic}
		\Input $\ma{\Psi}=(\vec{\psi_1|\ldots|\vec{\psi}_N})\in\mathbb{R}^{d\times N}$,  $\tilde{\ma{C}}\in\mathbb{R}^{d\times d}$ as defined in the text\\
		$\mathcal{S} \leftarrow \emptyset$\\
		Define $\vec{p}\in\mathbb{R}^N$ : $\forall i, \quad p(i) = \vec{\psi}_i^\adjoint\tilde{\ma{C}}\vec{\psi}_i$\\
		\textbf{for} $n=1,\ldots,k$ \textbf{do}:\\
		\hspace{0.5cm}$\bm{\cdot}$ Draw $s_n$ with proba $\mathbb{P}(s)=p(s)/\sum_{i}p(i)$\\
		\hspace{0.5cm}$\bm{\cdot}$ $\mathcal{S} \leftarrow \mathcal{S}\cup\{s_n\}$\\
		\hspace{0.5cm}$\bm{\cdot}$ Compute 
		$\vec{f}_n = \vec{\psi}_{s_n} - \sum_{l=1}^{n-1} \vec{f}_l(\vec{f}_l^\adjoint\tilde{\ma{C}}\vec{\psi}_{s_n})\in\mathbb{R}^{d}$\\
		\hspace{0.5cm}$\bm{\cdot}$ Normalize $\vec{f}_n \leftarrow \vec{f}_n / \sqrt{\vec{f}_n^\adjoint\tilde{\ma{C}}\vec{\psi}_{s_n}}$\\
		\hspace{0.5cm}$\bm{\cdot}$ Update $\vec{p}$ : $\forall i\quad p(i) \leftarrow p(i) - (\vec{f}_n^\adjoint\tilde{\ma{C}}\vec{\psi}_i)^2$\\
		\textbf{end for}
		\Output $\mathcal{S}$ of size $k$.
	\end{algorithmic}
\end{algorithm}

\begin{lemma}
	alg.~\ref{alg:sampling_DPP_efficient_dual} with inputs $\ma{\Psi}$ and $\tilde{\ma{C}}$ is equivalent to  alg.~\ref{alg:sampling_DPP_efficient} with input $\ma{V}=\ma{\Psi}^\adjoint\ma{W}\tilde{\ma{E}}^{-1/2}$.
\end{lemma}

\begin{proof}
	We show a point-wise equivalence. Given Eq.~\eqref{eq:yi}, the initial value of $p(i)$ is the same in both algorithms:
	\begin{align}
		\norm{\vec{y}_i}^2 = \vec{\psi}_i^\adjoint\ma{W}\tilde{\ma{E}}^{-1}\ma{W}^\adjoint\vec{\psi}_i = \vec{\psi}_i^\adjoint \tilde{\ma{C}}\vec{\psi}_i.
	\end{align}
	More generally, the following holds for all $(i,j)$:
	\begin{align}
	\label{eq:gene}
	\vec{y}_i^\adjoint \vec{y}_j= \vec{\psi}_i^\adjoint \tilde{\ma{C}}\vec{\psi}_j.
	\end{align}
	To differentiate notations, let us write $\tilde{\vec{f}}_n$ the vectors $\vec{f}_n$ defined in alg.~\ref{alg:sampling_DPP_efficient_dual}, and keep the notation $\vec{f}_n$ for the ones defined in alg.~\ref{alg:sampling_DPP_efficient}. 
	We then need to show that the values of $p(i)$ are the same in both algorithm, that is:
	\begin{align}
	\forall i, n \qquad \tilde{\vec{f}}_n^\adjoint\tilde{\ma{C}}\vec{\psi}_i = \vec{f}_n^\adjoint\vec{y}_i.
	\end{align}
	We prove this with a recurrence. 
	\textit{Initialization.} It is true for $n=1$:
	\begin{align}
	 \tilde{\vec{f}}_1^\adjoint\tilde{\ma{C}}\vec{\psi}_i = \frac{\vec{\psi}_{s_1}^\adjoint\tilde{\ma{C}}\vec{\psi}_i}{\sqrt{\vec{\psi}_{s_1}^\adjoint\tilde{\ma{C}}\vec{\psi}_{s_1}}}=\frac{\vec{y}_{s_1}^\adjoint\vec{y}_i}{\sqrt{\vec{y}_{s_1}^\adjoint\vec{y}_{s_1}}}= \vec{f}_1^\adjoint\vec{y}_i. 
	\end{align}
	\emph{Hypothesis.} We assume it true for all integers strictly inferior to $n$. \emph{Recurrence.} Let us show it is true for $n$. We have:
	\begin{align}
	\tilde{\vec{f}}_n^\adjoint\tilde{\ma{C}}\vec{\psi}_i = \frac{\vec{\psi}_{s_n}^\adjoint\tilde{\ma{C}}\vec{\psi}_i - \sum_{l=1}^{n-1}(\tilde{\vec{f}}_l^\adjoint\tilde{\ma{C}}\vec{\psi}_i)(\tilde{\vec{f}}_l^\adjoint\tilde{\ma{C}}\vec{\psi}_{s_n})}{\sqrt{\vec{\psi}_{s_n}^\adjoint\tilde{\ma{C}}\vec{\psi}_{s_n}-\sum_{l=1}^{n-1}(\tilde{\vec{f}}_l^\adjoint\tilde{\ma{C}}\vec{\psi}_{s_n})^2}}, 
	\end{align}
	such that, by hypothesis, and using Eq.~\eqref{eq:gene}
	\begin{align}
	\tilde{\vec{f}}_n^\adjoint\tilde{\ma{C}}\vec{\psi}_i = \frac{\vec{y}_{s_n}^\adjoint\vec{y}_i - \sum_{l=1}^{n-1}(\vec{f}_l^\adjoint\vec{y}_i)(\vec{f}_l^\adjoint\vec{y}_{s_n})}{\sqrt{\vec{y}_{s_n}^\adjoint\vec{y}_{s_n}-\sum_{l=1}^{n-1}(\vec{f}_l^\adjoint\vec{y}_{s_n})^2}}=\vec{f}_n^\adjoint\vec{y}_i,
	\end{align}
	thus ending the proof.
\end{proof}

\subsection{Implementation details}

Numerical difficulties can creep in when $\mu$ is large, regardless of which algorithm is used. From a quick glance at alg. \ref{alg:sampling_DPP_efficient_dual}, we see that $\vec{p}$ decreases at every step, and when implemented in finite precision it may go negative. We suggest setting it to zero manually at indices already sampled, and setting small negative values to 0. The difficulties are greatest when sampling DPPs that are highly repulsive, in which case one may not be able to sample more than a few points (unless one is willing to implement the algorithm in multiple precision arithmetic, which is rather slow to run). 

Depending on the programming language and the linear algebra backend, it may be much more efficient to implement the following step:
\begin{equation}
  \label{eq:expensive-step}
  \vec{f}_n = \vec{y}_{s_n} - \sum_{l=1}^{n-1} \vec{f}_l(\vec{f}_l^\adjoint\vec{y}_{s_n})
\end{equation}
as two matrix multiplications rather than a loop over $l$. Simply stack $\vec{f}_{1} \ldots \vec{f}_{n-1}$ in a matrix  $\mathbf{F}_{n-1}$, and compute the equivalent form:
\begin{equation}
  \label{eq:expensive-step-matrix}
  \vec{f}_n = \vec{y}_{s_n} - \mathbf{F}_{n-1} \mathbf{F}_{n-1}^{t} \vec{y}_{s_n}
\end{equation}
One benefit of the above formulation is that it can take advantage of a parallel BLAS (and indeed, this is the only step in alg. \ref{alg:sampling_DPP_efficient} where parallelisation could be used in a meaningful way). 

\section{Conclusion}

We have described two algorithms for exact sampling of discrete DPPs. For low-rank $L$-ensembles, we have found in practice that well-implemented exact algorithms are often competitive with approximate samplers like the Gibbs sampler \cite{li_fast_2016}. The great challenge for both types of algorithms lies in scaling in $\mu$: it is easy enough to sample from very large datasets ($N$ in the millions), but because of the quadratic scaling in $\mu$ one is limited to taking small samples ($\mu$ in the hundreds). Because increasing $\mu$ also means increasing the rank of the L-matrix, one ends up paying a double penalty: when forming the $L$-ensemble, and during sampling.

One solution would be to formulate a sparse $L$-ensemble, in which case alg. \ref{alg:sampling_DPP_efficient} carries over in sparse matrix algebra. Another, perhaps simpler approach is to subdivide the dataset ahead of time and use DPPs in each subset. We plan to investigate these possibilities in future work.